\documentclass[12pt]{article}
\usepackage{latexsym, amsmath, amssymb}
\usepackage{graphicx}
\usepackage{epstopdf}
\usepackage{amsmath, amsthm, amscd, amsfonts}
\usepackage{amsmath, xcolor}
\usepackage{array}
\usepackage{multirow}
\usepackage{stfloats}
\usepackage{amssymb}
\usepackage{ltablex}
\usepackage{float}
\usepackage{cite}
\usepackage{enumerate}
\usepackage{placeins}
\usepackage{comment}
 \usepackage{soul}
\usepackage{color}
\usepackage{lscape}
\usepackage{soul}
\usepackage{caption} 
\usepackage[a4paper, total={7in, 10in}]{geometry}
 
\usepackage[colorlinks,urlcolor=blue]{hyperref}

\newtheorem{theorem}{Theorem}[section]
\newtheorem{proposition}{Proposition}[section]
\newtheorem{lemma}{Lemma}[section]

\newtheorem{definition}{Definition}[section]
\newtheorem{example}{Example}[section]

\makeatletter
\newcommand*{\rom}[1]{\expandafter\@slowromancap\romannumeral #1@}
\makeatother

\def\bege{\begin{equation}} \def\ende{\end{equation}}

   \def\begr{\begin{eqnarray}}
\def\endr{\end{eqnarray}} 
 
\def\bege{\begin{equation}} \def\ende{\end{equation}}
\def\begr{\begin{eqnarray}} \def\endr{\end{eqnarray}} \def\bnum{\begin{enumerate}} \def\enum{\end{enumerate}}
\begin{document}
	
\begin{center}\Large
\textbf{Hamming and Symbol-Pair Distances of Constacyclic Codes of Length $2p^s$ over $\frac{\mathbb{F}_{p^m}[u, v]}{\langle u^2, v^2, uv-vu\rangle}$}
\end{center}

\begin{center}
Divya Acharya$^{1}$, Prasanna Poojary$^{1, \ast}$, Vadiraja Bhatta G R$^{2}$
\end{center}

\begin{center}
$^1$Department of Mathematics, Manipal Institute of Technology Bengaluru, Manipal Academy of Higher Education, Manipal, Karnataka, India\\ 
$^2$ Department of Mathematics, Manipal Institute of Technology, Manipal Academy of Higher Education, Manipal, Karnataka, India\\ 
\end{center}
\abstract{Let $p$ be an odd prime. In this paper, we have determined the Hamming distances for constacyclic codes of length $2p^s$  over the finite commutative non-chain ring $\mathcal{R}=\frac{\mathbb{F}_{p^m}[u, v]}{\langle u^2, v^2, uv-vu\rangle}$. Also their symbol-pair distances are completely obtained.}
$^{\ast}$Corresponding: prasanna.poojary@manipal.edu; poojaryprasanna34@gmail.com	
\section{Introduction}
	The ability to efficiently encode constacyclic codes with a simple shift register makes them valuable.
 In the finite field $\mathbb{F}_q$, $\alpha$-constacyclic codes of length $n$ over $\mathbb{F}_q$ are categorized as ideals $\langle \ell(x)\rangle$  of the ambient ring $\frac{\mathbb{F}_q[x]}{\langle x^n-\alpha \rangle}$ , where $\ell(x)$ is a divisor of $x^n-\alpha$. The codes are referred to as simple root codes when the length of the code $n$ is relatively prime to the characteristic of the finite field $\mathbb{F}_q$. If not, they are known as repeated-root codes, which were initially studied in \cite{berman1967semisimple,massey1973polynomial,falkner1979existence,roth1986cyclic,van1991repeated, castagnoli1991repeated}.
 
 The structure of $\lambda$-constacyclic codes 
 over the chain ring $\frac{\mathbb{F}_{p^m}[u]}{\langle u^k \rangle}$ was examined by Dinh et al. \cite{dinh2016repeated} and Guenda and Gulliver \cite{guenda2015repeated}.
The algebraic structures of all constacyclic codes of length $mp^s$ over the ring $\mathbb{F}_{p^m}+u\mathbb{F}_{p^m}$, where $m=1,2,3,4$, are obtained in a series of papers  \cite{dinh2010constacyclic, dinh2015negacyclic, chen2016constacyclic, dinh2020constacyclic3ps, dinh2018negacyclic, dinh2018cyclic, dinh2019alpha,dinh2019class}. In \cite{yildiz2011cyclic}, Yildiz and Karadeniz examined cyclic codes of odd length over the ring 
$\frac{\mathbb{F}_{2}[u, v]}{\langle u^2, v^2, uv-vu\rangle}$ which is not a chain ring. As the Gray images of these cyclic codes, they found a few good binary codes. In \cite{karadeniz20111+}  Yildiz and Karadeniz studied $(1+v)$-constacyclic codes over the ring $\mathbb{F}_{2}+u\mathbb{F}_{2}+v\mathbb{F}_{2}+uv\mathbb{F}_{2}$ and they obtained cyclic codes over $\mathbb{F}_{2}+u\mathbb{F}_{2}$ as the image of  $(1+v)$-constacyclic codes over $\mathbb{F}_{2}+u\mathbb{F}_{2}+v\mathbb{F}_{2}+uv\mathbb{F}_{2}$ under natural Gray map.
	In \cite{yu20141} Haifeng et al. examined  $(1-uv)$-constacyclic codes over $\mathbb{F}_{p}+u\mathbb{F}_{p}+v\mathbb{F}_{p}+uv\mathbb{F}_{p}$ and showed that under a Gray map, images of these codes are distance invariant quasi-cyclic code of length $p^3n$ and index $p^2$ over $\mathbb{F}_{p}$. Negacyclic codes of odd length over $\frac{\mathbb{F}_{p}[u, v]}{\langle u^2, v^2, uv-vu\rangle}$ are studied by Ghosh \cite{ghosh2015negacyclic}. Bag \cite{bag2019classes} investigated $(\lambda_1 + u\lambda_2)$ and $(\lambda_1 + v\lambda_3)$-constacyclic codes of prime power length over $\frac{\mathbb{F}_{p^m}[u, v]}{\langle u^2, v^2, uv-vu\rangle}$.
	The more general ring $\frac{\mathbb{F}_{2}[u_1,u_2,\ldots,u_k]}{\langle u_{i}^2, v_{j}^2, u_iv_j-v_ju_i\rangle}$ was taken into consideration by the authors of \cite{dougherty2012cyclic}, who also examined the general characteristics of cyclic codes over such rings and identified the nontrivial one-generator cyclic codes.
	These investigations were expanded to cyclic codes over the ring $\frac{\mathbb{F}_{2^m}[u, v]}{\langle u^2, v^2, uv-vu\rangle}$ by Sobhani and Molakarimi in \cite{sobhani2013some}. Dinh et al. in \cite{dinh2020constacyclicps} investigated the structures of all constacyclic codes of prime power length over the ring $\frac{\mathbb{F}_{p^m}[u, v]}{\langle u^2, v^2, uv-vu\rangle}$.
	
One of the main problems in coding theory is determining the Hamming distance. In 2010, Dinh \cite{dinh2010constacyclic} computed the Hamming distances of all ($\alpha+u \beta$)-constacyclic codes of length $p^s$ over $\mathbb{F}_{p^m}+u\mathbb{F}_{p^m}$. Later, for all $\gamma$ -constacyclic codes of length $p^s$ over $\mathbb{F}_{p^m}+u\mathbb{F}_{p^m}$ Hamming distances are studied in \cite{dinh2018hamming}. The Hamming distances of all constacyclic codes of length $2p^s$	over $\mathbb{F}_{p^m}+u\mathbb{F}_{p^m}$ are completely determined in \cite{dinh2020hamming}. In 2020, the Hamming distance of $(\alpha+u \beta)$-constacyclic codes ( where $\alpha+u \beta$ is not a cube) of length $3p^s$ over $\mathbb{F}_{p^m}+u\mathbb{F}_{p^m}$ is established in \cite{dinh2020constacyclic3ps}. For $\lambda$-constacyclic codes of length $3p^s$ over $\mathbb{F}_{p^m}+u\mathbb{F}_{p^m}$, where $\lambda =\alpha+u \beta$ is a cube and $\lambda$ is not a cube in $\mathbb{F}_{p^m}$ Hamming distance was determined in \cite{dinh2023hamming}. Hamming distances of $\lambda$-constacyclic codes for $\lambda = \alpha+ \beta u +\delta uv, \alpha+\gamma v+\delta uv, \alpha +\beta u +\gamma v + \delta uv$ where $\alpha , \beta, \gamma \in \mathbb{F}^*_{p^m}$ and $\delta \in \mathbb{F}_{p^m}$ are completely determined in \cite{dinh2022self}.
	
A novel coding method for symbol pair read channels was presented by Cassuto et al. \cite{40} in 2010. In this method, the outputs of the read process are consecutive symbol pairs. The symbol-pair  distance of constacyclic codes over $\mathbb{F}_{p^m}$ are studied in \cite{43, dinh2019symbol, 51}. Symbol-Pair distances of repeated-root constacyclic codes of prime power lengths over $\mathbb{F}_{p^m}+u\mathbb{F}_{p^m}$ and  $\mathbb{F}_{p^m} [u]/\langle u^3\rangle$ were studied in \cite{dinh2018hamming} and \cite{44} respectively. Also, Maximum Distance Separable Ssymbol-pair codes were studied by several authors(see, for example \cite{45,46,47,48,49,50} )

Section 2 provides some preliminary information, and the remaining portion of the paper is structured as follows. Section 3 gives the brief outline of $(\alpha_1 + \alpha_2 u + \alpha_3 v + \alpha_4 uv)$-constacyclic code of length $2p^s$ over the ring $\mathcal{R}$ and also we determined the number of codewords. In section 4, we obtain the Hamming distances of $(\alpha_1 +\alpha_2 u+  \alpha_3 v + \alpha_4 uv)$ and $(\alpha_1 +\alpha_3 v + \alpha_4 uv)$-constacyclic codes of length $2p^s$ over $\mathcal{R}$. In section 5, we determine the symbol-pair  distance of $\alpha$-constacyclic codes of length $2p^s$ over $\mathcal{R}$, where $\alpha =\alpha_1 +\alpha_2 u+  \alpha_3 v + \alpha_4 uv, \alpha_1 +\alpha_3 v + \alpha_4 uv$. 

\section{Preliminaries}
	Let $\Re$ be a finite commutative ring with identity 1 and let $I$ be an ideal. If only one element generates $I$ of $\Re$,  it is referred to as a $principal ideal$. If every ideal in the ring $\Re$ is a principle,  then the ring is a $principal~ ideal~ ring$. If $\Re$ has a unique maximum ideal,  then $\Re$ is referred to as a $local~ ring$. If the set of all ideals in  $\Re$ forms a chain under inclusion,  the ring is referred to as a $chain~ ring.$
	\begin{proposition}\cite{dinh2004cyclic}\label{prop1}
		For a finite commutative ring $\Re$ the following conditions are equivalent:
		\begin{enumerate}
			\item $\Re$ is a local ring and the maximal ideal of is principal;
			\item $\Re$ is a local principal ideal ring;
			\item $\Re$ is a chain ring.
		\end{enumerate}
	\end{proposition}
	Let $\Re$ be a finite commutative ring with identity. A code $\mathcal{C}$ of length $n$ over $\Re$ is a nonempty subset of $\Re^n$. An element of $\mathcal{C}$ is called a codeword. If $\mathcal{C}$ is an $\Re$-submodule of $\Re^n$,  then $\mathcal{C}$ is said to be linear. Let $\alpha$ be a unit of $\Re$.  The $\alpha$-constacyclic shift $\sigma_{\alpha}$ on $\Re^n$ is the shift,
	
	\begin{center}
		$\sigma_{\alpha}(\zeta_0,  \zeta_1,\ldots,  \zeta_{n-1}) = (\alpha \zeta_{n-1},  \zeta_0,\ldots, \zeta_{n-2}).$
	\end{center}
	and a code $\mathcal{C}$ is said to be $\alpha$-constacyclic if $\mathcal{C}$ is closed under the $\alpha$-constacyclic shift $\sigma_{\alpha}$.
	If  $\alpha$ is equal to 1(or -1), then the $\alpha$-constacyclic codes are referred to as cyclic (or negacyclic) codes.

	Each codeword $\zeta = (\zeta_0,  \zeta_1, \ldots,  \zeta_{n-1})$ in $\mathcal{C}$ can be viewed with its polynomial representation $\zeta(x) = \zeta_ 0+ \zeta_1x +\cdots+
	\zeta_{n-1}x^{n-1}$.
	Each codeword in $\mathcal{C}$, denoted by $\zeta = (\zeta_0,  \zeta_1, \ldots,  \zeta_{n-1})$, uniquely identified using its polynomial representation, $\zeta(x) = \zeta_ 0+ \zeta_1x +\cdots+
	\zeta_{n-1}x^{n-1}$. From that, the well-known \cite{macwilliams1977theory, huffman2010fundamentals} and simple proposition as follows:
	\begin{proposition}\label{1}
		A linear code $\mathcal{C}$ of length $n$ over $\Re$ is an $\alpha$-constacyclic if and only if $\mathcal{C}$ is an ideal of $\frac{\Re[x]}{\langle x^n-\alpha \rangle}$.
	\end{proposition}
The Hamming weight of $x$, represented as $wt_H(x)$, is the nonzero entries of a codeword $x = (x_0, x_1,\ldots x_{n-1})\in A^n$. The Hamming distance $d_H(x,y)$ of two
words $x$ and $y$ equals the number of components in which they differ. i.e., $d_H(x,y)=wt_H(x-y)$. 
The Hamming distance of code $\mathcal{C}$ containing at least two words, $d(\mathcal{C}) = min\{d(x, y): x, y \in \mathcal{C}, x \neq y\}$. The Hamming weight of $\mathcal{C}$, denoted $wt(\mathcal{C})$, is the smallest of the weights of the nonzero codewords of $\mathcal{C}$. For a linear code, the Hamming weight and Hamming distance of the code are equal.

The symbol-pair distance was defined by Cassuto and Blaum in \cite{40}.
 
    Let $x = (x_0, x_1,\ldots x_{n-1})$ and $y = (y_0, y_1,\ldots y_{n-1})$  be two vectors in $A^n$, where $A$ is a code alphabet. The symbol-pair distance is given by using Hamming distance over the alphabet $(A, A)$ as follows:
    \begin{equation*}
        d_{sp}(x, y) = \vert \{i : (x_i, x_{i+1})\neq (y_i, y_{i+1})\}\vert.
     \end{equation*}
Then  $d_{sp}(\mathcal{C}) =\underset{\underset{x\neq y}{x,y\in \mathcal{C}}}{min}\{d_{sp}(x, y)\} $ is the symbol-pair distance of code $\mathcal{C}$.

	Let $p$ be an odd prime number and $m$ be a positive integer. Let $\mathcal{R} = \frac{\mathbb{F}_{p^m}[u, v]}{\langle u^2, v^2, uv-vu\rangle}=\mathbb{F}_{p^m} + u\mathbb{F}_{p^m} + v\mathbb{F}_{p^m} +uv\mathbb{F}_{p^m}$ ($u^2 = 0$,  $v^2 = 0$,  $uv = vu$). An element $\alpha=\alpha_1 + \alpha_2 u + \alpha_3 v + \alpha_4 uv \in \mathcal{R}$ is a unit iff $\alpha_1$ is non zero in $\mathbb{F}_{p^m}$. Let $\alpha_1 \in \mathbb{F}^{*}_{p^m}$. Now $\alpha_1^{p^{tm}}=\alpha_1$ for any positive integer $t$. For positive integer $m$ and $s$, by Division algorithm,  there exists non-negative integers $q_0$ and $r_0$ such that $s= q_0 m+r_0$ with $0\leq r_0 \leq m-1$. Let $\alpha_0=\alpha_1^{p^{(q_0+1)m-s}}=\alpha_1^{p^{m-r_0}}$. Then ${\alpha_0}^{p^{s}}=\alpha_1^{p^{(q_0+1)m}}=\alpha_1$.

    For any unit $\alpha$ of $\mathcal{R}$, let 
  \begin{center}
        $\mathcal{R}_\alpha=\frac{\mathcal{R}[x]}{\langle x^{2p^s}-\alpha \rangle}$.
  \end{center}
  It follows from Proposition \ref{1} that $\alpha$-constacyclic codes of length $2p^s$ over $\mathcal{R}$ are ideals
of $\mathcal{R}_\alpha$.

\section{$(\alpha_1 + \alpha_2 u + \alpha_3 v + \alpha_4 uv)$-constacyclic code of length $2p^s$ over the ring $\mathcal{R}$}
In \cite{41}, Y. Ahendouz and I. Akharraz studied the $(\alpha_1 + \alpha_2 u + \alpha_3 v + \alpha_4 uv)$-constacyclic code of length $2p^s$ over the ring $\mathcal{R}$, where $\alpha_1  \in \mathbb{F}^*_{p^m}$ and $\alpha_2, \alpha_3,\alpha_4 \in \mathbb{F}_{p^m}$ with $\alpha_2$ and $\alpha_3$ not both zero. Also, their duals are determined. 

	Let $\psi_u$ denote reduction modulo $u$ map from $R_{\alpha_1+\alpha_{2}u+  \alpha_{3}v+\alpha_{4} uv}$ to $\frac{(\mathbb{F}_{p^m} + v\mathbb{F}_{p^m})[x]}{\langle x^{2p^s}-(\alpha_1 + \alpha_3 v) \rangle}$. The map $\psi_u$ is a surjective ring homomorphism. Let
	$\mathcal{C}$ be an ideal of $R_{\alpha_1+\alpha_{2}u+  \alpha_{3}v+\alpha_{4} uv}$. We define the torsion of $\mathcal{C}$ by $Tor(\mathcal{C})=\{a(x) \in \frac{(\mathbb{F}_{p^m} + v\mathbb{F}_{p^m})[x]}{\langle x^{2p^s}-(\alpha_1 + \alpha_3 v) \rangle} : ua(x)\in \mathcal{C}\}$ and residue of $\mathcal{C}$ by  $Res(\mathcal{C})=\psi_u(\mathcal{C}).$ Clearly, $Tor(\mathcal{C})$ and $Res(\mathcal{C})$  are ideals of $\frac{(\mathbb{F}_{p^m} + v\mathbb{F}_{p^m})[x]}{\langle x^{2p^s}-(\alpha_1 + \alpha_3 v) \rangle}$ and $Tor(\mathcal{C}) \cong Ker(\psi_u\vert \mathcal{C} )$. Therefore, $\vert \mathcal{C}\vert = \vert Tor(\mathcal{C})\vert \vert Res(\mathcal{C})\vert $. We calculate the number of codewords in $Tor(\mathcal{C})$ and $Res(\mathcal{C})$ to get the number of codewords in $\mathcal{C}$.  We know from \cite{chen2016constacyclic} that any ideal of $\frac{(\mathbb{F}_{p^m} + v\mathbb{F}_{p^m})[x]}{\langle x^{2p^s}-(\alpha_1 + \alpha_3 v) \rangle}$ is of the form  $\langle (x^2-\alpha_0)^\ell \rangle $, where
	$0 \leq \ell \leq 2p^s$ and it has $p^{2m(2p^s-\ell)}$ codewords.

The following theorem gives the number of codewords, $\eta_\mathcal{C}$, of the $\alpha$-constacyclic code of length $2p^s$ over the ring $\mathcal{R}$. For computational purposes, we incorporate the results of \cite{41} into the cases.



\begin{theorem}\label{thm1}
    Let $\alpha$ be a unit of the ring $\mathcal{R}$. Then,
    \begin{enumerate}
        \item Suppose $\alpha$ is a square in $\mathcal{R}$, say $\alpha=\gamma^2$, where $\gamma \in \mathcal{R}$. Then
    \begin{equation*}
        \frac{\mathcal{R}[x]}{\langle x^{2p^s}-\alpha \rangle} \cong \frac{\mathcal{R}[x]}{\langle x^{p^s}+\gamma  \rangle}\oplus\frac{\mathcal{R}[x]}{\langle x^{p^s}-\gamma \rangle}.
    \end{equation*}
    i.e., a direct sum of an $\gamma$-constacyclic code and $-\gamma$-constacyclic code of length $p^s$ over $\mathcal{R}$  can be used to represent each $\alpha$-constacyclic code of length $2p^s$.
        \item Suppose $\alpha=\alpha_1 + \alpha_2 u + \alpha_3 v + \alpha_4 uv$ is not a square in $\mathcal{R}$, where $\alpha_1 , \alpha_2, \alpha_3 \in \mathbb{F}^*_{p^m}$ and $\alpha_4 \in \mathbb{F}_{p^m}$ the distinct ideals of the ring $R_{\alpha_1+\alpha_{2}u+  \alpha_{3}v+\alpha_{4} uv}$ are given by
        \begin{itemize}
        \item Type A: $\langle 0 \rangle $, $\langle 1 \rangle $. And $\eta_\mathcal{C}=1$ and $\eta_\mathcal{C}=p^{8mp^{s}}$ respectively.
        \item Type B: $\langle u(x^2-\alpha_0)^\ell \rangle $,where $0\leq  \ell \leq 2p^s-1$ And $\eta_\mathcal{C}=p^{2m(2p^{s}-\ell)}$.
        \item Type C: $\langle (x^2-\alpha_0)^\ell +u(x^2-\alpha_0)^t z(x) \rangle $, where $0\leq  \ell \leq 2p^s-1$,  $0\leq  t < \ell $ and either $z(x)$ is 0 or $z(x)$ is a unit which can be represented as $z(x)=\sum\limits_{\kappa}^{}(z_{0\kappa}x+z_{1\kappa})(x^2-\alpha_0)^{\kappa}$ with $z_{0\kappa},  z_{1\kappa} \in \mathbb{F}_{p^m}$ and $z_{00}x+z_{10} \neq 0$. And \begin{itemize}
				\item If $z(x)=0$, then
				\begin{center}
                    $\eta_\mathcal{C}=$
                    $\begin{cases}
                         p^{4m(2p^s-\ell)} & \text{if}\quad 1\leq \ell\leq p^s;\\
                        p^{2m(3p^s-\ell)} & \text{if}\quad p^s < \ell \leq 2p^s-1.
                    \end{cases}$
				\end{center}
				\item If $z(x)$ is a unit and $\ell \neq p^s+t$ then 
				\begin{center}
					$\eta_\mathcal{C}=$
					$\begin{cases}
						p^{4m(2p^s-\ell)} & \text{if}\quad 1\leq \ell\leq p^s;\\
						p^{2m(3p^s-\ell)} & \text{if}\quad p^s \leq \ell < p^s+t;\\
						p^{2m(2p^s-t)} & \text{if} \quad  p^s+ t < \ell \leq 2p^s-1.
					\end{cases}$
				\end{center} 
				\item If $z(x)$ is a unit and $\ell=p^s+t$, then $\eta_\mathcal{C}=p^{2m(2p^s-t)}$.
			\end{itemize}
        \item Type D: 
				$\langle (x^2-\alpha_0)^\ell +u\sum\limits_{\kappa=0}^{\mu-1}(a_{\kappa}x+b_{\kappa})(x^2-\alpha_0)^{\kappa}, u (x^2-\alpha_0)^{\mu} \rangle $, where $0\leq  \ell \leq 2p^s-1$,  $a_{\kappa}, b_{\kappa} \in  \mathbb{F}_{p^m}$ and $\mu\leq \Im$ with $\Im$ being the smallest integer such that 
            \begin{center}
				$u (x^2-\alpha_0)^{\Im} \in \langle (x^2-\alpha_0)^\ell +u\sum\limits_{\kappa=0}^{\mu-1}(a_{\kappa}x+b_{\kappa})(x^2-\alpha_0)^{\kappa}  \rangle $
            \end{center}
			or equivalently, 
			\begin{center}
				$\langle (x^2-\alpha_0)^\ell +u(x^2-\alpha_0)^t z(x),  u (x^2-\alpha_0)^{\mu} \rangle $
			\end{center}
			with $z(x)$ as in Type C and deg $(z(x))\leq\mu-t-1$.
   and
			\begin{center}
				$\mu < \Im=$
				$\begin{cases}
					min\{\ell, p^s\}  & \text{if}\quad z(x)=0;\\
					min\{\ell,p^s,  2p^s-\ell+t\}  &\text{if}\quad z(x)\neq 0\quad \text{and}\quad \ell \neq p^s+t;\\
					p^s &\text{if}\quad z(x)\neq 0 \quad\text{and}\quad \ell= p^s+t.
				\end{cases}$
			\end{center}
	And $\eta_\mathcal{C}=p^{2m(4p^s-\ell-\mu)}$.
		\end{itemize}
    
    \item Suppose $\alpha=\alpha_1 + \alpha_3 v + \alpha_4 uv $ be a non-square in  $\mathcal{R}$ where $\alpha_1 , \alpha_3, \alpha_4 \in \mathbb{F}^*_{p^m}$ the distinct ideals of the ring $R_{\alpha_1+\alpha_{3}v+\alpha_{4} uv}$ are given by
    \begin{itemize}
        \item Type A: $\langle 0 \rangle $, $\langle 1 \rangle $. And $\eta_\mathcal{C}=1$ and $\eta_\mathcal{C}=p^{8mp^{s}}$ respectively.
        \item Type B: $\langle u(x^2-\alpha_0)^\ell \rangle $,where $0\leq  \ell \leq 2p^s-1$ And $\eta_\mathcal{C}=p^{2m(2p^{s}-\ell)}$.
        \item Type C: $\langle (x^2-\alpha_0)^\ell +u(x^2-\alpha_0)^t z(x) \rangle $, where $0\leq  \ell \leq 2p^s-1$,  $0\leq  t < \ell $ and either $z(x)$ is 0 or $z(x)$ is a unit which can be represented as $z(x)=\sum\limits_{\kappa}^{}(z_{0\kappa}x+z_{1\kappa})(x^2-\alpha_0)^{\kappa}$ with $z_{0\kappa},  z_{1\kappa} \in \mathbb{F}_{p^m}$ and $z_{00}x+z_{10} \neq 0$. And \begin{itemize}
				\item If $z(x)=0$, then
			$\eta_\mathcal{C}=p^{4m(2p^s-\ell)}.$
				\item If $z(x)$ is a unit and $\ell \neq p^s+t$ then 
				\begin{center}
            $\eta_\mathcal{C}=$
            $\begin{cases}
                p^{4m(2p^s-\ell)} & \text{if}\quad 1\leq \ell\leq p^s+\frac{t}{2};\\
                p^{2m(2p^s-t)} & \text{if} \quad  p^s+ \frac{t}{2} < \ell \leq 2p^s-1.
            \end{cases}$
        \end{center} 
				\item If $z(x)$ is a unit and $\ell=p^s+t$, then $\eta_\mathcal{C}=p^{2m(2p^s-t)}$.
			\end{itemize}
        \item Type D: 
				$\langle (x^2-\alpha_0)^\ell +u\sum\limits_{\kappa=0}^{\mu-1}(a_{\kappa}x+b_{\kappa})(x^2-\alpha_0)^{\kappa}, u (x^2-\alpha_0)^{\mu} \rangle $, where $0\leq  \ell \leq 2p^s-1$,  $a_{\kappa}, b_{\kappa} \in  \mathbb{F}_{p^m}$ and $\mu\leq \Im$ with $\Im$ being the smallest integer such that 
            \begin{center}
				$u (x^2-\alpha_0)^{\Im} \in \langle (x^2-\alpha_0)^\ell +u\sum\limits_{\kappa=0}^{\mu-1}(a_{\kappa}x+b_{\kappa})(x^2-\alpha_0)^{\kappa}  \rangle $
            \end{center}
			or equivalently, 
			\begin{center}
				$\langle (x^2-\alpha_0)^\ell +u(x^2-\alpha_0)^t z(x),  u (x^2-\alpha_0)^{\mu} \rangle $
			\end{center}
			with $z(x)$ as in Type C and deg $(z(x))\leq\mu-t-1$.
   and
   \begin{center}
			$\mu < \Im=$
			$\begin{cases}
				\ell  & \text{if}\quad z(x)=0;\\
				min\{\ell, 2p^s +t-\ell\}  &\text{if}\quad z(x)\neq 0. \\
			\end{cases}$
		\end{center}
	And $\eta_\mathcal{C}=p^{2m(4p^s-\ell-\mu)}$.
		\end{itemize}
    \end{enumerate}
\end{theorem}

\section{Hamming distance}
    In general, it is quite challenging to calculate the exact values of Hamming distances of a class of codes. For the class of  $(\alpha_1 +\alpha_2 u+  \alpha_3 v + \alpha_4 uv)$-constacyclic code of length $2p^s$ over $R_{u^2, v^2, p^m}$  when  $\alpha=(\alpha_1 +\alpha_2 u+  \alpha_3 v + \alpha_4 uv)$ is not a square in $\mathbb{F}_{p^m}$ , we calculate the Hamming distances. Here, we revisit these findings from \cite{dinh2020hamming}.
    \begin{theorem}\cite{dinh2020hamming}\label{thm3}
        Let $\mathcal{C}=\langle (x^2-\alpha_0)^\ell \rangle $$\subseteq$ $\frac{\mathbb{F}_{p^m}[x]}{\langle x^{2p^s}-\alpha \rangle}$,  for $0 \leq \ell \leq p^s$.  Then the  Hamming distance $d_H(\mathcal{C})$ is completely given by
        \begin{center}
            $d_H(\mathcal{C})=$
            $\begin{cases}
                1  &\text{if}\quad \ell=0;\\
                 2  &\text{if}\quad 1\leq  \ell \leq p^{s-1};\\
                (\beta_{0} +2)  &\text{if}\quad \beta_{0} p^{s-1}+1\leq \ell\leq (\beta_{0}+1)p^{s-1} \quad where \quad 1 \leq \beta_{0} \leq p-2;\\
                (\Gamma+1)p^{\gamma}  & \text{if}\quad p^s-p^{s-\gamma}+(\Gamma+1)p^{s-\gamma-1}+1\leq \ell\leq p^s-p^{s-\gamma}+\Gamma p^{s-\gamma-1}\\
                & \quad where\quad 1 \leq \Gamma \leq p-1  \quad and \quad 1 \leq \gamma \leq s-1;\\
                 0  &\text{if}\quad \ell=p^s.
            \end{cases}$
        \end{center}
    \end{theorem}
    The Hamming distance of Type A ideals $\langle 0 \rangle$ and $\langle 1 \rangle$ of $R_{\alpha_1+\alpha_{2}u+  \alpha_{3}v+\alpha_{4} uv}$ are obviously 0 and 1,  respectively. We calculate the Hamming distance of Type B, C and Type D codes in the following theorems.
\begin{theorem}\label{thm4}
    Let $\mathcal{C}=\langle u(x^2-\alpha_0)^\ell \rangle $,  where $0\leq  \ell \leq 2p^s-1$. Then the  Hamming distance of $\mathcal{C}$, $d_H(\mathcal{C})$ is given by
    \begin{center}
        $d_H(\mathcal{C})=$
        $\begin{cases}
            1  &\text{if}\quad 0\leq  \ell \leq p^s;\\
            (\Gamma+1)p^{\gamma}& \text{if} \quad 2p^s-p^{s-\gamma }+(\Gamma+1)p^{s-\gamma-1}+1\leq \ell\leq 2p^s-p^{s-\gamma}+\Gamma p^{s-\gamma-1}; \\
            &\text{where} \quad 1 \leq \Gamma \leq p-1, \quad \text{and} \quad 0 \leq \gamma \leq s-1.
        \end{cases}$
    \end{center}
    \end{theorem}
    \begin{proof}
        \begin{itemize}
            \item \textbf{Case(i): }If $0\leq  \ell \leq p^s$ then $d_H(\mathcal{C})=1$.
            \item  \textbf{Case(ii): }If $p^s+1 \leq \ell \leq 2p^s-1$,  $\mathcal{C}=\langle u(x^2-\alpha_0)^\ell \rangle =\langle uv(x^2-\alpha_0)^\ell-{p^s }\rangle $. Thus,  $\mathcal{C}$ is precisely the $(\alpha_1 +\alpha_2 u+ \alpha_3 v + \alpha_4 uv)$-constacyclic code $\langle (x^2-\alpha_0)^{\ell-p^s} \rangle$ in $\frac{\mathbb{F}_{p^m}[x]}{\langle x^{2p^s}-(\alpha_1 +\alpha_2 u+ \alpha_3 v + \alpha_4 uv) \rangle}$ multiplied by $uv$. Therefore $d_H(\mathcal{C})= d_H(\langle (x^2-\alpha_0)^{\kappa}\rangle_{\mathbb{F}_{p^m}})$ and proof follows from Theorem \ref{thm3}.
        \end{itemize}
    \end{proof}
    \begin{theorem}\label{thm5}
        Let $\mathcal{C}=\langle (x^2-\alpha_0)^\ell +u(x^2-\alpha_0)^t z(x) \rangle $,  where $0\leq  \ell \leq 2p^s-1$,  $0\leq  t < \ell  $ and $z(x)$ is 0 or a unit. Then
        \begin{center}
        $d_H(\mathcal{C})=$
        $\begin{cases}
            1  &\text{if}\quad 0\leq  \Im \leq p^s;\\
            (\Gamma+1)p^{\gamma}& \text{if} \quad 2p^s-p^{s-\gamma }+(\Gamma+1)p^{s-\gamma-1}+1\leq \Im \leq 2p^s-p^{s-\gamma}+\Gamma p^{s-\gamma-1}; \\
            &\text{where} \quad 1 \leq \Gamma \leq p-1, \quad \text{and} \quad 0 \leq \gamma \leq s-1.
        \end{cases}$
    \end{center}
    \end{theorem}
    \begin{proof}
        Let $\Im$ be the smallest integer such that $u(x^2-\alpha_0)^\Im \in \mathcal{C}$. Then $d_H(\mathcal{C})\leq d_H(\langle u(x^2-\alpha_0)^\Im \rangle)$. Let $a(x) \in \mathcal{C}$. Then there exists $b(x)=\sum\limits_{\kappa=0}^{2p^s-1}(a_{0\kappa}x+b_{0\kappa})(x^2-\alpha_0)^{\kappa}+u\sum\limits_{\kappa=0}^{2p^s-1}(a_{1\kappa}x+b_{1\kappa})(x^2-\alpha_0)^{\kappa} \in R_{\alpha_1 , \alpha_2,  \alpha_3,  \alpha_4}$, where $a_{0\kappa},  b_{0\kappa},  a_{1\kappa},  b_{1\kappa} \in \mathbb{F}_{p^m}$ such that $a(x)=b(x)[(x^2-\alpha_0)^\ell+u(x^2-\alpha_0)^t z(x)]$. Thus
        \begin{align}\label{eqn1}
            a(x)=&
            \Bigg[
                \notag \sum\limits_{\kappa=0}^{2p^s-1}(a_{0\kappa}x+b_{0\kappa})(x^2-\alpha_0)^{\kappa}+u\sum\limits_{\kappa=0}^{2p^s-1}(a_{1\kappa}x+b_{1\kappa})(x^2-\alpha_0)^{\kappa}
			\Bigg]\\
			\notag& \times [(x^2-          \alpha_0)^\ell          +u(x^2-\alpha_0)^t z(x)] \\
			\notag=&(x^2-\alpha_0)^\ell \sum\limits_{\kappa=0}^{2p^s-1}(a_{0\kappa}x+b_{0\kappa})(x^2-\alpha_0)^{\kappa}\\
			\notag&+u(x^2-\alpha_0)^\ell \sum\limits_{\kappa=0}^{2p^s-1}(a_{1\kappa}x+b_{1\kappa})(x^2-\alpha_0)^{\kappa}\\
			&+u(x^2-\alpha_0)^t z(x)\sum\limits_{\kappa=0}^{2p^s-1}(a_{0\kappa}x+b_{0\kappa})(x^2-\alpha_0)^{\kappa}.
		\end{align}
  
Then 
\begin{align*}
    w_H(a(x))&\geq w_H(u a(x))\\
    &=w_H\Bigg(u (x^2-          \alpha_0)^\ell 
 \sum\limits_{\kappa=0}^{2p^s-1}(a_{0\kappa}x+b_{0\kappa})(x^2-\alpha_0)^{\kappa} \Bigg)\\
 &\geq d_H\Bigg(\Big
 \langle u (x^2-          \alpha_0)^\ell \Big \rangle \Bigg)\\
 &\geq d_H\Bigg(\Big
 \langle u (x^2-          \alpha_0)^\Im \Big \rangle \Bigg)\qquad (\text{since} \quad \Im \leq \ell)
\end{align*}

  Suppose $a(x)\neq 0$ and $ua(x)=0$, since nilpotency index of $(x^2-\alpha_0)$ is $3p^s$, $a_{0\kappa}=b_{0\kappa}=0$ for $0\leq \kappa \leq 2p^s-\ell -1$. From Equation \ref{eqn1}
		\begin{align*}
			a(x)=&2\beta u(x^2-\alpha_0)^{p^s}
			\sum\limits_{\kappa=0}^{\ell-1}(a_{0\kappa}x+b_{0\kappa})(x^2-\alpha_0)^{\kappa}\\
			&+u(x^2-\alpha_0)^\ell \sum\limits_{\kappa=0}^{2p^s-\ell-1}(a_{1\kappa}x+b_{1\kappa})(x^2-\alpha_0)^{\kappa}\\
			&+u(x^2-\alpha_0)^{2p^s+t-\ell} z(x)\sum\limits_{\kappa=0}^{\ell-1}(a_{0\kappa}x+b_{0\kappa})(x^2-\alpha_0)^{\kappa}.
		\end{align*}
		By Theorem \ref{thm1}, we have $a(x)\in \langle u(x^2-\alpha_0)^\Im \rangle$. Then $d_H(\langle u(x^2-\alpha_0)^\Im \rangle)\leq d_H(\mathcal{C}) $. Hence $d_H(\mathcal{C})= d_H(\langle u(x^2-\alpha_0)^\Im \rangle)$.  The proof follows from Theorem \ref{thm4}.
  
	\end{proof} 
	\begin{theorem}\label{thm6}
		Let $\mathcal{C}=\langle (x^2-\alpha_0)^\ell +u(x^2-\alpha_0)^t z(x),  u (x^2-\alpha_0)^{\mu} \rangle $,  where $0\leq  \ell \leq 2p^s-1$,  $0\leq  t < \ell  $ and $z(x)$ is 0 or $z(x)$ is a unit. Then the  Hamming distance of $\mathcal{C}$,  $d_H(\mathcal{C})$ is given by
		\begin{center}
        $d_H(\mathcal{C})=$
        $\begin{cases}
            1  &\text{if}\quad 0\leq  \mu \leq p^s;\\
            (\Gamma+1)p^{\gamma}& \text{if} \quad 2p^s-p^{s-\gamma }+(\Gamma+1)p^{s-\gamma-1}+1\leq \mu \leq 2p^s-p^{s-\gamma}+\Gamma p^{s-\gamma-1}; \\
            &\text{where} \quad 1 \leq \Gamma \leq p-1, \quad \text{and} \quad 0 \leq \gamma \leq s-1.
        \end{cases}$
    \end{center}
	\end{theorem}
	\begin{proof}
		We have $u(x^2-\alpha_0)^\mu \in \mathcal{C}$. Then $d_H(\mathcal{C})\leq d_H(\langle u(x^2-\alpha_0)^\mu \rangle)$. Let $a(x) \in \mathcal{C}$. Then there exists $b(x)=\sum\limits_{\kappa=0}^{2p^s-1}(a_{0\kappa}x+b_{0\kappa})(x^2-\alpha_0)^{\kappa}+u\sum\limits_{\kappa=0}^{2p^s-1}(a_{1\kappa}x+b_{1\kappa})(x^2-\alpha_0)^{\kappa} \in R_{\alpha_1, \alpha_2,  \alpha_3,  \alpha_4} $,  where $a_{0\kappa},  b_{0\kappa},  a_{1\kappa},  b_{1\kappa} \in \mathbb{F}_{p^m}$ and $b^\prime(x)=\sum\limits_{\kappa=0}^{2p^s-1}(a^\prime_{0\kappa}x+b^\prime_{0\kappa})(x^2-\alpha_0)^{\kappa}+u\sum\limits_{\kappa=0}^{2p^s-1}(a^\prime_{1\kappa}x+b^\prime_{1\kappa})(x^2-\alpha_0)^{\kappa}$, where $a^\prime_{0\kappa},  b^\prime_{0\kappa},  a^\prime_{1\kappa},  b^\prime_{1\kappa} \in \mathbb{F}_{p^m}$, such that $a(x) =b(x)[(x^2-\alpha_0)^\ell+u(x^2-\alpha_0)^t z(x)]+b^\prime(x)u(x^2-\alpha_0)^\mu$. Thus
		\begin{align}\label{eqn2}
			\notag a(x)=&
			\Bigg[ \sum\limits_{\kappa=0}^{2p^s-1}(a_{0\kappa}x+b_{0\kappa})(x^2-\alpha_0)^{\kappa}+u\sum\limits_{\kappa=0}^{2p^s-1}(a_{1\kappa}x+b_{1\kappa})(x^2-\alpha_0)^{\kappa}
			\Bigg ] [(x^2-\alpha_0)^\ell +u(x^2-\alpha_0)^t z(x)]\\
			\notag &+ \Bigg[
				\sum\limits_{\kappa=0}^{2p^s-1}(a^\prime_{0\kappa}x+b^\prime_{0\kappa})(x^2-\alpha_0)^{\kappa}+u\sum\limits_{\kappa=0}^{2p^s-1}(a^\prime_{1\kappa}x+b^\prime_{1\kappa})(x^2-\alpha_0)^{\kappa}
			\Bigg] u(x^2-\alpha_0)^\mu\\
			\notag =&(x^2-\alpha_0)^\ell \sum\limits_{\kappa=0}^{2p^s-1}(a_{0\kappa}x+b_{0\kappa})(x^2-\alpha_0)^{\kappa}\\
			\notag &+u(x^2-\alpha_0)^\ell \sum\limits_{\kappa=0}^{2p^s-1}(a_{1\kappa}x+b_{1\kappa})(x^2-\alpha_0)^{\kappa}\\
			 \notag&+u(x^2-\alpha_0)^t z(x)\sum\limits_{\kappa=0}^{2p^s-1}(a_{0\kappa}x+b_{0\kappa})(x^2-\alpha_0)^{\kappa}\\
			&+u(x^2-\alpha_0)^\mu\sum\limits_{\kappa=0}^{2p^s-1}(a^\prime_{0\kappa}x+b^\prime_{0\kappa})(x^2-\alpha_0)^{\kappa}.
		\end{align}
  Then 
\begin{align*}
    w_H(a(x))&\geq w_H(u a(x))\\
    &=w_H\Bigg(u (x^2-          \alpha_0)^\ell 
 \sum\limits_{\kappa=0}^{2p^s-1}(a_{0\kappa}x+b_{0\kappa})(x^2-\alpha_0)^{\kappa} \Bigg)\\
 &\geq d_H\Bigg(\Big
 \langle u (x^2-          \alpha_0)^\ell \Big \rangle \Bigg)\\
 &\geq d_H\Bigg(\Big
 \langle u (x^2-          \alpha_0)^\mu \Big \rangle \Bigg),\qquad as \quad\mu <\Im \leq \ell
\end{align*}
  Suppose $a(x)\neq 0$ and $ua(x)=0$, since nilpotency index of $(x^2-\alpha_0)$ is $3p^s$, $a_{0\kappa}=b_{0\kappa}=0$ for $0\leq \kappa \leq 2p^s-\ell -1$. From Equation \ref{eqn2}
		\begin{align*}
			a(x)=&2\beta u(x^2-\alpha_0)^{p^s}
			\sum\limits_{\kappa=0}^{\ell-1}(a_{0\kappa}x+b_{0\kappa})(x^2-\alpha_0)^{\kappa}\\
			&+u(x^2-\alpha_0)^\ell \sum\limits_{\kappa=0}^{2p^s-\ell-1}(a_{1\kappa}x+b_{1\kappa})(x^2-\alpha_0)^{\kappa}\\
			&+u(x^2-\alpha_0)^{2p^s+t-\ell} z(x)\sum\limits_{\kappa=0}^{\ell-1}(a_{0\kappa}x+b_{0\kappa})(x^2-\alpha_0)^{\kappa}\\	&+u(x^2-\alpha_0)^\mu\sum\limits_{\kappa=0}^{2p^s-1}(a^\prime_{0\kappa}x+b^\prime_{0\kappa})(x^2-\alpha_0)^{\kappa}
		\end{align*}
		We have $\mu <\Im$ and by Theorem \ref{thm1},  $a(x)\in \langle u(x^2-\alpha_0)^\mu \rangle$. Then $d_H(\langle u(x^2-\alpha_0)^\mu \rangle) \leq d_H(\mathcal{C}) $.  Hence $d_H(\mathcal{C})= d_H(\langle u(x^2-\alpha_0)^\mu \rangle)$. The proof follows from Theorem \ref{thm4}.
	\end{proof}
	The Hamming distance of Type A ideals $\langle 0 \rangle$ and $\langle 1 \rangle$ of $R_{\alpha_1+ \alpha_{3}v+\alpha_{4} uv}$ are obviously 0 and 1,  respectively. We calculate the Hamming distance of Type B, C and Type D codes in the following theorems.

	\begin{theorem}\label{thm7}
		Let $\mathcal{C}=\langle u(x^2-\alpha_0)^\ell \rangle $,  where $0\leq  \ell \leq 2p^s-1$. Then the  Hamming distance of $\mathcal{C}$, $d_H(\mathcal{C})$ is given by
		\begin{center}
			$d_H(\mathcal{C})=$
			$\begin{cases}
				1  &\text{if}\quad 0\leq  \ell \leq p^s;\\
				(\Gamma+1)p^{\gamma}& \text{if} \quad 2p^s-p^{s-\gamma }+(\Gamma-1)p^{s-\gamma-1}+1\leq \ell\leq 2p^s-p^{s-\gamma}+\Gamma p^{s-\gamma-1}; \\
				&\text{where} \quad 1 \leq \Gamma \leq p-1, \quad \text{and} \quad 0 \leq \gamma \leq s-1.
			\end{cases}$
		\end{center}
	\end{theorem}
	\begin{proof}
 
It follows the same steps as Theorem \ref{thm4}

	\end{proof}
	\begin{theorem}\label{thm8}
		Let $\mathcal{C}=\langle (x^2-\alpha_0)^\ell +u(x^2-\alpha_0)^t z(x) \rangle $,  where $0\leq  \ell \leq 2p^s-1$,  $0\leq  t < \ell  $ and $z(x)$ is 0 or a unit. Then
\begin{center}
    $d_H(\mathcal{C})=$
    $\begin{cases}
        1  &\text{if}\quad 0\leq  \Im \leq p^s;\\
        (\Gamma+1)p^{\gamma}& \text{if} \quad 2p^s-p^{s-\gamma }+(\Gamma-1)p^{s-\gamma-1}+1\leq \Im \leq 2p^s-p^{s-\gamma}+\Gamma p^{s-\gamma-1}; \\
        &\text{where} \quad 1 \leq \Gamma \leq p-1, \quad \text{and} \quad 0 \leq \gamma \leq s-1.
    \end{cases}$
\end{center}
\end{theorem}
\begin{proof}
It follows the same steps as Theorem \ref{thm5}.
\end{proof}

\begin{theorem}\label{thm10}
    Let $\mathcal{C}=\langle (x^2-\alpha_0)^\ell +u(x^2-\alpha_0)^t z(x),  u (x^2-\alpha_0)^{\mu} \rangle $,  where $0\leq  \ell \leq 2p^s-1$,  $0\leq  t < \ell  $ and $z(x)$ is 0 or $z(x)$ is a unit. Then the  Hamming distance of $\mathcal{C}$,  $d_H(\mathcal{C})$ is given by
    
   \begin{center}
        $d_H(\mathcal{C})=$
        $\begin{cases}
            1  &\text{if}\quad 0\leq  \mu \leq p^s;\\
            (\Gamma+1)p^{\gamma}& \text{if} \quad 2p^s-p^{s-\gamma }+(\Gamma-1)p^{s-\gamma-1}+1\leq \mu \leq 2p^s-p^{s-\gamma}+\Gamma p^{s-\gamma-1}; \\
            &\text{where} \quad 1 \leq \Gamma \leq p-1, \quad \text{and} \quad 0 \leq \gamma \leq s-1.
        \end{cases}$
    \end{center}
\end{theorem}
	\begin{proof}
 It follows the same steps as Theorem \ref{thm6}.
 \end{proof}

\section{Symbol-Pair Distance}
 For the class of  $(\alpha_1 +\alpha_2 u+  \alpha_3 v + \alpha_4 uv)$-constacyclic code of length $2p^s$ over $\mathcal{R}$  when  $\alpha=(\alpha_1 +\alpha_2 u+  \alpha_3 v + \alpha_4 uv)$ is not a square in $\mathbb{F}_{p^m}$ , we calculate the symbol- pair distances. Here, we revisit these findings from \cite{dinh2019symbol}.
    \begin{theorem}\cite{dinh2019symbol}\label{thm11}
        Let $\mathcal{C}=\langle (x^2-\alpha_0)^\ell \rangle $$\subseteq$ $\frac{\mathbb{F}_{p^m}[x]}{\langle x^{2p^s}-\alpha \rangle}$,  for $0 \leq \ell \leq p^s$.  Then the symbol-pair distance $d_{sp}(\mathcal{C})$ is completely given by
        \begin{center}
            $d_{sp}(\mathcal{C})=$
            $\begin{cases}
                2  &\text{if}\quad \ell=0;\\
                2(\Gamma+1)p^{\gamma}  & \text{if}\quad p^s-p^{s-\gamma}+(\Gamma-1)p^{s-\gamma-1}+1\leq \ell\leq p^s-p^{s-\gamma}+\Gamma p^{s-\gamma-1}\\
                & \quad where\quad 1 \leq \Gamma \leq p-1  \quad and \quad 1 \leq \gamma \leq s-1;\\
                 0  &\text{if}\quad \ell=p^s.
            \end{cases}$
        \end{center}
    \end{theorem}
    The symbol-pair distance of Type A ideals $\langle 0 \rangle$ and $\langle 1 \rangle$ of $R_{\alpha_1, \alpha_2,  \alpha_3, \alpha_4}$ are obviously 0 and 2,  respectively. We calculate the symbol-pair distance of Type B, C and Type D codes in the following theorems.
\begin{theorem}\label{thm12}
    Let $\mathcal{C}=\langle u(x^2-\alpha_0)^\ell \rangle $,  where $0\leq  \ell \leq 2p^s-1$. Then the  minimum symbol-pair distance of $\mathcal{C}$, $d_{sp}(\mathcal{C})$ is given by
    \begin{center}
        $d_{sp}(\mathcal{C})=$
        $\begin{cases}
            2  &\text{if}\quad  0 \leq \ell \leq p^s;\\
            2(\Gamma+1)p^{\gamma}& \text{if} \quad 2p^s-p^{s-\gamma }+(\Gamma-1)p^{s-\gamma-1}+1\leq \ell\leq 2p^s-p^{s-\gamma}+\Gamma p^{s-\gamma-1}; \\
            &\text{where} \quad 1 \leq \Gamma \leq p-1, \quad \text{and} \quad 0 \leq \gamma \leq s-1.
        \end{cases}$
    \end{center}
    \end{theorem}
    \begin{proof}
        \begin{itemize}
            \item \textbf{Case(i): }If $ 0 \leq \ell \leq p^s$ then $d_{sp}(\mathcal{C})=2$.
            \item  \textbf{Case(ii): }If $\quad 2p^s-p^{s-\gamma }+(\Gamma-1)p^{s-\gamma-1}+1\leq \ell\leq 2p^s-p^{s-\gamma}+\Gamma p^{s-\gamma-1}$,  $\mathcal{C}=\langle u(x^2-\alpha_0)^\ell \rangle =\langle uv(x^2-\alpha_0)^{\ell-p^s }\rangle $. Thus,  $\mathcal{C}$ is precisely the $(\alpha_1 +\alpha_2 u+ \alpha_3 v + \alpha_4 uv)$-constacyclic code $\langle (x^2-\alpha_0)^{\ell-p^s} \rangle$ in $\frac{\mathbb{F}_{p^m}[x]}{\langle x^{2p^s}-(\alpha_1 +\alpha_2 u+ \alpha_3 v + \alpha_4 uv) \rangle}$ multiplied by $uv$. Therefore $d_{sp}(\mathcal{C})= d_{sp}(\langle (x^2-\alpha_0)^{\kappa}\rangle_{\mathbb{F}_{p^m}})$ and proof follows from Theorem \ref{thm11}.
        \end{itemize}
    \end{proof}

    \begin{theorem}\label{thm13}
        Let $\mathcal{C}=\langle (x^2-\alpha_0)^\ell +u(x^2-\alpha_0)^t z(x) \rangle $,  where $0\leq  \ell \leq 2p^s-1$,  $0\leq  t < \ell  $ and $z(x)$ is 0 or a unit. Then
       \begin{center}
        $d_{sp}(\mathcal{C})=$
        $\begin{cases}
            2  &\text{if}\quad  0 \leq \Im \leq p^s;\\
            2(\Gamma+1)p^{\gamma}& \text{if} \quad 2p^s-p^{s-\gamma }+(\Gamma-1)p^{s-\gamma-1}+1\leq \Im\leq 2p^s-p^{s-\gamma}+\Gamma p^{s-\gamma-1}; \\
            &\text{where} \quad 1 \leq \Gamma \leq p-1, \quad \text{and} \quad 0 \leq \gamma \leq s-1.
        \end{cases}$
    \end{center}
    \end{theorem}
    \begin{proof}
        Let $\Im$ be the smallest integer such that $u(x^2-\alpha_0)^\Im \in \mathcal{C}$. Then $d_{sp}(\mathcal{C})\leq d_{sp}(\langle u(x^2-\alpha_0)^\Im \rangle)$. Let $a(x) \in \mathcal{C}$. Then there exists $b(x)=\sum\limits_{\kappa=0}^{2p^s-1}(a_{0\kappa}x+b_{0\kappa})(x^2-\alpha_0)^{\kappa}+u\sum\limits_{\kappa=0}^{2p^s-1}(a_{1\kappa}x+b_{1\kappa})(x^2-\alpha_0)^{\kappa} \in R_{\alpha_1 , \alpha_2,  \alpha_3,  \alpha_4}$, where $a_{0\kappa},  b_{0\kappa},  a_{1\kappa},  b_{1\kappa} \in \mathbb{F}_{p^m}$ such that $a(x)=b(x)[(x^2-\alpha_0)^\ell+u(x^2-\alpha_0)^t z(x)]$. Thus
        \begin{align}\label{eqn3}
           \notag a(x)=&
            \Bigg[ \sum\limits_{\kappa=0}^{2p^s-1}(a_{0\kappa}x+b_{0\kappa})(x^2-\alpha_0)^{\kappa}+u\sum\limits_{\kappa=0}^{2p^s-1}(a_{1\kappa}x+b_{1\kappa})(x^2-\alpha_0)^{\kappa}
			\Bigg]\\
			\notag & \times \Bigg[(x^2-          \alpha_0)^\ell          +u(x^2-\alpha_0)^t z(x)\Bigg] \\
		\notag	=&(x^2-\alpha_0)^\ell \sum\limits_{\kappa=0}^{2p^s-1}(a_{0\kappa}x+b_{0\kappa})(x^2-\alpha_0)^{\kappa}\\
		\notag	&+u(x^2-\alpha_0)^\ell \sum\limits_{\kappa=0}^{2p^s-1}(a_{1\kappa}x+b_{1\kappa})(x^2-\alpha_0)^{\kappa}\\
			&+u(x^2-\alpha_0)^t z(x)\sum\limits_{\kappa=0}^{2p^s-1}(a_{0\kappa}x+b_{0\kappa})(x^2-\alpha_0)^{\kappa}.
		\end{align}
Then 
\begin{align*}
    w_{sp}(a(x))&\geq w_{sp}(u a(x))\\
    &=w_{sp}\Bigg(u (x^2-          \alpha_0)^\ell 
 \sum\limits_{\kappa=0}^{2p^s-1}(a_{0\kappa}x+b_{0\kappa})(x^2-\alpha_0)^{\kappa} \Bigg)\\
 &\geq d_{sp}\Bigg(\Big
 \langle u (x^2-          \alpha_0)^\ell \Big \rangle \Bigg)\\
 &\geq d_{sp}\Bigg(\Big
 \langle u (x^2-          \alpha_0)^\Im \Big \rangle \Bigg),\qquad \text{as} \quad \Im \leq \ell.
\end{align*}

  Suppose $a(x)\neq 0$ and $ua(x)=0$, since nilpotency index of $(x^2-\alpha_0)$ is $3p^s$, $a_{0\kappa}=b_{0\kappa}=0$ for $0\leq \kappa \leq 2p^s-\ell -1$. From Equation \ref{eqn3}

		\begin{align*}
			a(x)=&2\beta u(x^2-\alpha_0)^{p^s}
			\sum\limits_{\kappa=0}^{\ell-1}(a_{0\kappa}x+b_{0\kappa})(x^2-\alpha_0)^{\kappa}\\
			&+u(x^2-\alpha_0)^\ell \sum\limits_{\kappa=0}^{2p^s-\ell-1}(a_{1\kappa}x+b_{1\kappa})(x^2-\alpha_0)^{\kappa}\\
			&+u(x^2-\alpha_0)^{2p^s+t-\ell} z(x)\sum\limits_{\kappa=0}^{\ell-1}(a_{0\kappa}x+b_{0\kappa})(x^2-\alpha_0)^{\kappa}.
		\end{align*}
		By Theorem \ref{thm1} we have $a(x)\in \langle u(x^2-\alpha_0)^\Im \rangle$. Then we get $d_{sp}(\langle u(x^2-\alpha_0)^\Im \rangle)\leq d_{sp}(\mathcal{C}) $. Hence $d_{sp}(\mathcal{C})= d_{sp}(\langle u(x^2-\alpha_0)^\Im \rangle)$.

	\end{proof} 
	\begin{theorem}\label{thm14}
		Let $\mathcal{C}=\langle (x^2-\alpha_0)^\ell +u(x^2-\alpha_0)^t z(x),  u (x^2-\alpha_0)^{\mu} \rangle $,  where $0\leq  \ell \leq 2p^s-1$,  $0\leq  t < \ell  $ and $z(x)$ is 0 or $z(x)$ is a unit. Then the  Symbol-Pair distance of $\mathcal{C}$,  $d_{sp}(\mathcal{C})$ is given by
		\begin{center}
        $d_{sp}(\mathcal{C})=$
        $\begin{cases}
            2  &\text{if}\quad  0 \leq \mu \leq p^s;\\
            2(\Gamma+1)p^{\gamma}& \text{if} \quad 2p^s-p^{s-\gamma }+(\Gamma-1)p^{s-\gamma-1}+1\leq \mu\leq 2p^s-p^{s-\gamma}+\Gamma p^{s-\gamma-1}; \\
            &\text{where} \quad 1 \leq \Gamma \leq p-1, \quad \text{and} \quad 0 \leq \gamma \leq s-1.
        \end{cases}$
    \end{center}
	\end{theorem}
	\begin{proof}
		We have $u(x^2-\alpha_0)^\mu \in \mathcal{C}$. Then $d_{sp}(\mathcal{C})\leq d_{sp}(\langle u(x^2-\alpha_0)^\mu \rangle)$. Let $a(x) \in \mathcal{C}$. Then there exists $b(x)=\sum\limits_{\kappa=0}^{2p^s-1}(a_{0\kappa}x+b_{0\kappa})(x^2-\alpha_0)^{\kappa}+u\sum\limits_{\kappa=0}^{2p^s-1}(a_{1\kappa}x+b_{1\kappa})(x^2-\alpha_0)^{\kappa} \in R_{\alpha_1, \alpha_2,  \alpha_3,  \alpha_4} $,  where $a_{0\kappa},  b_{0\kappa},  a_{1\kappa},  b_{1\kappa} \in \mathbb{F}_{p^m}$ and $b^\prime(x)=\sum\limits_{\kappa=0}^{2p^s-1}(a^\prime_{0\kappa}x+b^\prime_{0\kappa})(x^2-\alpha_0)^{\kappa}+u\sum\limits_{\kappa=0}^{2p^s-1}(a^\prime_{1\kappa}x+b^\prime_{1\kappa})(x^2-\alpha_0)^{\kappa}$, where $a^\prime_{0\kappa},  b^\prime_{0\kappa},  a^\prime_{1\kappa},  b^\prime_{1\kappa} \in \mathbb{F}_{p^m}$, such that $a(x) =b(x)[(x^2-\alpha_0)^\ell+u(x^2-\alpha_0)^t z(x)]+b^\prime(x)u(x^2-\alpha_0)^\mu$. Thus
		\begin{align}\label{eqn4}
			\notag a(x)=&
			\Bigg[
				\sum\limits_{\kappa=0}^{2p^s-1}(a_{0\kappa}x+b_{0\kappa})(x^2-\alpha_0)^{\kappa}+u\sum\limits_{\kappa=0}^{2p^s-1}(a_{1\kappa}x+b_{1\kappa})(x^2-\alpha_0)^{\kappa}
			\Bigg] [(x^2-\alpha_0)^\ell +u(x^2-\alpha_0)^t z(x)]\\
		\notag 	&+ \Bigg[
				\sum\limits_{\kappa=0}^{2p^s-1}(a^\prime_{0\kappa}x+b^\prime_{0\kappa})(x^2-\alpha_0)^{\kappa}+u\sum\limits_{\kappa=0}^{2p^s-1}(a^\prime_{1\kappa}x+b^\prime_{1\kappa})(x^2-\alpha_0)^{\kappa}
			\Bigg] u(x^2-\alpha_0)^\mu\\
		\notag 	=&(x^2-\alpha_0)^\ell \sum\limits_{\kappa=0}^{2p^s-1}(a_{0\kappa}x+b_{0\kappa})(x^2-\alpha_0)^{\kappa}\\
		\notag 	&+u(x^2-\alpha_0)^\ell \sum\limits_{\kappa=0}^{2p^s-1}(a_{1\kappa}x+b_{1\kappa})(x^2-\alpha_0)^{\kappa}\\
			\notag &+u(x^2-\alpha_0)^t z(x)\sum\limits_{\kappa=0}^{2p^s-1}(a_{0\kappa}x+b_{0\kappa})(x^2-\alpha_0)^{\kappa}\\
			&+u(x^2-\alpha_0)^\mu\sum\limits_{\kappa=0}^{2p^s-1}(a^\prime_{0\kappa}x+b^\prime_{0\kappa})(x^2-\alpha_0)^{\kappa}.
		\end{align}
		 Then 
\begin{align*}
    w_{sp}(a(x))&\geq w_{sp}(u a(x))\\
    &=w_{sp}\Bigg(u (x^2-          \alpha_0)^\ell 
 \sum\limits_{\kappa=0}^{2p^s-1}(a_{0\kappa}x+b_{0\kappa})(x^2-\alpha_0)^{\kappa} \Bigg)\\
 &\geq d_{sp}\Bigg(\Big
 \langle u (x^2-          \alpha_0)^\ell \Big \rangle \Bigg)\\
 &\geq d_{sp}\Bigg(\Big
 \langle u (x^2-          \alpha_0)^\mu \Big \rangle \Bigg),\qquad as \quad\mu <\Im \leq \ell
\end{align*}
  Suppose $a(x)\neq 0$ and $ua(x)=0$, since nilpotency index of $(x^2-\alpha_0)$ is $3p^s$, $a_{0\kappa}=b_{0\kappa}=0$ for $0\leq \kappa \leq 2p^s-\ell -1$. From Equation \ref{eqn4}
		\begin{align*}
			a(x)=&2\beta u(x^2-\alpha_0)^{p^s}
			\sum\limits_{\kappa=0}^{\ell-1}(a_{0\kappa}x+b_{0\kappa})(x^2-\alpha_0)^{\kappa}\\
			&+u(x^2-\alpha_0)^\ell \sum\limits_{\kappa=0}^{2p^s-\ell-1}(a_{1\kappa}x+b_{1\kappa})(x^2-\alpha_0)^{\kappa}\\
			&+u(x^2-\alpha_0)^{2p^s+t-\ell} z(x)\sum\limits_{\kappa=0}^{\ell-1}(a_{0\kappa}x+b_{0\kappa})(x^2-\alpha_0)^{\kappa}\\	&+u(x^2-\alpha_0)^\mu\sum\limits_{\kappa=0}^{2p^s-1}(a^\prime_{0\kappa}x+b^\prime_{0\kappa})(x^2-\alpha_0)^{\kappa}
		\end{align*}
		We have $\mu <\Im$ and by Proposition \ref{thm1},  $a(x)\in \langle u(x^2-\alpha_0)^\mu \rangle$. Then $d_{sp}(\langle u(x^2-\alpha_0)^\mu \rangle) \leq d_{sp}(\mathcal{C}) $.  Hence $d_{sp}(\mathcal{C})= d_{sp}(\langle u(x^2-\alpha_0)^\mu \rangle)$. The proof follows from Theorem \ref{thm4}.
	\end{proof}
	The Symbol-Pair distance of Type A ideals $\langle 0 \rangle$ and $\langle 1 \rangle$ of $R_{\alpha_1+\alpha_{3}v+\alpha_{4} uv}$ are obviously 0 and 2,  respectively. We calculate the Symbol-Pair distance of Type B, C and Type D codes in the following theorems.

	\begin{theorem}\label{thm15}
		Let $\mathcal{C}=\langle u(x^2-\alpha_0)^\ell \rangle $,  where $0\leq  \ell \leq 2p^s-1$. Then the  Symbol-Pair distance of $\mathcal{C}$, $d_{sp}(\mathcal{C})$ is given by
		\begin{center}
        $d_{sp}(\mathcal{C})=$
        $\begin{cases}
            2  &\text{if}\quad  0 \leq \ell \leq p^s;\\
            2(\Gamma+1)p^{\gamma}& \text{if} \quad 2p^s-p^{s-\gamma }+(\Gamma-1)p^{s-\gamma-1}+1\leq \ell\leq 2p^s-p^{s-\gamma}+\Gamma p^{s-\gamma-1}; \\
            &\text{where} \quad 1 \leq \Gamma \leq p-1, \quad \text{and} \quad 0 \leq \gamma \leq s-1.
        \end{cases}$
    \end{center}
	\end{theorem}
	\begin{proof}
It follows the same steps as Theorem \ref{thm12}.
	\end{proof}
	\begin{theorem}\label{thm16}
		Let $\mathcal{C}=\langle (x^2-\alpha_0)^\ell +u(x^2-\alpha_0)^t z(x) \rangle $,  where $0\leq  \ell \leq 2p^s-1$,  $0\leq  t < \ell  $ and $z(x)$ is 0 or a unit. Then
\begin{center}
        $d_{sp}(\mathcal{C})=$
        $\begin{cases}
            2  &\text{if}\quad  0 \leq \Im \leq p^s;\\
            2(\Gamma+1)p^{\gamma}& \text{if} \quad 2p^s-p^{s-\gamma }+(\Gamma-1)p^{s-\gamma-1}+1\leq \Im\leq 2p^s-p^{s-\gamma}+\Gamma p^{s-\gamma-1}; \\
            &\text{where} \quad 1 \leq \Gamma \leq p-1, \quad \text{and} \quad 0 \leq \gamma \leq s-1.
        \end{cases}$
    \end{center}
\end{theorem}
\begin{proof}
It follows the same steps as Theorem \ref{thm13}.
\end{proof}

\begin{theorem}\label{thm17}
    Let $\mathcal{C}=\langle (x^2-\alpha_0)^\ell +u(x^2-\alpha_0)^t z(x),  u (x^2-\alpha_0)^{\mu} \rangle $,  where $0\leq  \ell \leq 2p^s-1$,  $0\leq  t < \ell  $ and $z(x)$ is 0 or $z(x)$ is a unit. Then the  Symbol-Pair distance of $\mathcal{C}$,  $d_{sp}(\mathcal{C})$ is given by
    
   \begin{center}
        $d_{sp}(\mathcal{C})=$
        $\begin{cases}
            2  &\text{if}\quad  0 \leq \mu \leq p^s;\\
            2(\Gamma+1)p^{\gamma}& \text{if} \quad 2p^s-p^{s-\gamma }+(\Gamma-1)p^{s-\gamma-1}+1\leq \mu\leq 2p^s-p^{s-\gamma}+\Gamma p^{s-\gamma-1}; \\
            &\text{where} \quad 1 \leq \Gamma \leq p-1, \quad \text{and} \quad 0 \leq \gamma \leq s-1.
        \end{cases}$
    \end{center}
\end{theorem}
	\begin{proof}
 It follows the same steps as Theorem \ref{thm14}.
	\end{proof}
\begin{example}
Let $p=3$, $m=1$, $s=1$ and $\alpha=2+v+uv$. It is easy to check that 2 is not a square in $\mathbb{F}_{3}$. We determine all Hamming and Symbol-Pair distance of $(2+v+uv)-$Constacyclic code of length 6 over $\mathbb{F}_{3} + u\mathbb{F}_{3} + v\mathbb{F}_{3} +uv\mathbb{F}_{3}$ in Table \ref{Tab1}

\begin{table}[H]
\centering
\caption{Hamming and Symbol-Pair distance of $(2+v+uv)-$Constacyclic code of length 6 over $\mathbb{F}_{3} + u\mathbb{F}_{3} + v\mathbb{F}_{3} +uv\mathbb{F}_{3}$}
\label{Tab1}
\begin{tabular}{|c|c|c|c|}
\hline              
Ideal($\mathcal{C})$&$\eta_{\mathcal{C}}$  &$d_H$&$d_{sp}$\\
 \hline
\textbf{Type A}&  & &\\
$\langle 0 \rangle$& 1 & 0&0\\
$\langle 1 \rangle$& $3^{24}$ & 1&2\\
\hline
\textbf{Type B}&  & &\\
$\langle u \rangle$& $3^{12}$ &1 &2\\
$\langle u(x^2-2) \rangle$& $3^{10}$ &1 &2\\
$\langle u(x^2-2)^2 \rangle$&  $3^{8}$&1 &2\\
$\langle u(x^2-2)^3 \rangle$& $3^{6}$ &1 &2\\
$\langle u(x^2-2)^4 \rangle$& $3^{4}$ &2 &4\\
$\langle u(x^2-2)^5 \rangle$& $3^{2}$ &3 &6\\
\hline
\textbf{Type C}&  & &\\
$\langle (x^2-2) \rangle$& $3^{20}$ & 1&2\\
$\langle (x^2-2)^2 \rangle$& $3^{16}$ & 1&2\\
$\langle (x^2-2)^3 \rangle$&$3^{12}$  &1 &2\\
$\langle (x^2-2)^4 \rangle$& $3^{8}$ & 2&4\\
$\langle (x^2-2)^5 \rangle$& $3^{4}$ &3 &6\\
\hline
$\langle (x^2-2)+uz(x) \rangle$& $3^{20}$ &1 &2\\
$\langle (x^2-2)^2+uz(x) \rangle$&  $3^{16}$&1 &2\\
$\langle (x^2-2)^2+u(x^2-2)z(x) \rangle$& $3^{16}$ &1 &2\\
$\langle (x^2-2)^3+uz(x) \rangle$& $3^{12}$ & 1&2\\
$\langle (x^2-2)^3+u(x^2-2)z(x) \rangle$& $3^{12}$ &1 &2\\
$\langle (x^2-2)^3+u(x^2-2)^2z(x) \rangle$& $3^{12}$ &1 &2\\
$\langle (x^2-2)^4 +uz(x)\rangle$& $3^{12}$ & 1&2\\
$\langle (x^2-2)^4+u(x^2-2)z(x) \rangle$& $3^{10}$ & 1&2\\
$\langle (x^2-2)^4+u(x^2-2)^2z(x) \rangle$& $3^{8}$  &2 &4\\
$\langle (x^2-2)^4+u(x^2-2)^3 z(x) \rangle$& $3^{8}$& 2&4\\
$\langle (x^2-2)^5+uz(x) \rangle$&  $3^{12}$ & 1&2\\
$\langle (x^2-2)^5+u(x^2-2)z(x) \rangle$&$3^{10}$ & 1&2\\
$\langle (x^2-2)^5+u(x^2-2)^2z(x) \rangle$& $3^{8}$ & 1&2\\
$\langle (x^2-2)^5+u(x^2-2)^3 z(x) \rangle$& $3^{6}$ & 2&4\\
$\langle (x^2-2)^5+u(x^2-2)^4 z(x) \rangle$& $3^{4}$ &3&6\\
\hline

\hline
\end{tabular}
\end{table}

\begin{table}[H]
\ContinuedFloat
    \centering
    \caption{(Continued.) Hamming and Symbol-Pair distance of $(2+v+uv)-$Constacyclic code of length 6 over $\mathbb{F}_{3} + u\mathbb{F}_{3} + v\mathbb{F}_{3} +uv\mathbb{F}_{3}$}
    \begin{tabular}{|c|c|c|c|}
\hline              
Ideal($\mathcal{C})$&$\eta_{\mathcal{C}}$  &$d_H$&$d_{sp}$\\
 \hline
 \textbf{Type D}&  & &\\
 $\langle (x^2-2),u \rangle$& $3^{22}$ & 1&2\\
$\langle (x^2-2)^2, u \rangle$& $3^{20}$ & 1&2\\
$\langle (x^2-2)^2, u(x^2-2) \rangle$& $3^{18}$ &1 &2\\
$\langle (x^2-2)^3, u \rangle$& $3^{18}$ & 1&2\\
$\langle (x^2-2)^3, u(x^2-2) \rangle$&$3^{16}$  & 1&2\\
$\langle (x^2-2)^3, u(x^2-2)^2 \rangle$&  $3^{14}$&1 &2\\
$\langle (x^2-2)^4,u \rangle$& $3^{16}$ & 1&2\\
$\langle (x^2-2)^4, u(x^2-2) \rangle$&  $3^{14}$&1 &2\\
$\langle (x^2-2)^4, u(x^2-2)^2 \rangle$&$3^{12}$  & 1&2\\
$\langle (x^2-2)^4, u(x^2-2)^3 \rangle$&$3^{10}$  & 1&2\\
$\langle (x^2-2)^5,u \rangle$& $3^{14}$ & 1&2\\
$\langle (x^2-2)^5, u(x^2-2) \rangle$&  $3^{12}$&1 &2\\
$\langle (x^2-2)^5, u(x^2-2)^2 \rangle$&$3^{10}$  & 1&2\\
$\langle (x^2-2)^5, u(x^2-2)^3 \rangle$& $3^{8}$ & 1&2\\
$\langle (x^2-2)^5, u(x^2-2)^4 \rangle$& $3^{6}$ &2 &4\\
\hline
$\langle (x^2-2)+uz(x),u \rangle$& $3^{22}$ & 1&\\
\hline
$\langle (x^2-2)^2+uz(x),u \rangle$& $3^{20}$ & 1&2\\
$\langle (x^2-2)^2+uz(x),u(x^2-2) \rangle$& $3^{18}$ &1 &2\\
\hline
$\langle (x^2-2)^2+u(x^2-2)z(x),u \rangle$& $3^{20}$ &1 &2\\
$\langle (x^2-2)^2+u(x^2-2)z(x),u (x^2-2)\rangle$& $3^{18}$ &1 &2\\
\hline
$\langle (x^2-2)^3+uz(x),u \rangle$&$3^{18}$  &1 &2\\
$\langle (x^2-2)^3+uz(x),u (x^2-2)\rangle$&$3^{16}$  & 1&2\\
$\langle (x^2-2)^3+uz(x),u (x^2-2)^2\rangle$& $3^{14}$ & 1&2\\
\hline
$\langle (x^2-2)^3+u(x^2-2)z(x),u  \rangle$& $3^{18}$ & 1&2\\
$\langle (x^2-2)^3+u(x^2-2)z(x),u  (x^2-2) \rangle$&  $3^{16}$&1 &2\\
$\langle (x^2-2)^3+u(x^2-2)z(x),u  (x^2-2)^2 \rangle$&$3^{14}$  &1 &2\\
\hline
$\langle (x^2-2)^3+u(x^2-2)^2z(x), u \rangle$& $3^{18}$ &1 &2\\
$\langle (x^2-2)^3+u(x^2-2)^2z(x), u (x^2-2) \rangle$& $3^{16}$ &1 &2\\
$\langle (x^2-2)^3+u(x^2-2)^2z(x), u (x^2-2)^2\rangle$& $3^{14}$ & 1&2\\
\hline
$\langle (x^2-2)^4 +uz(x), u\rangle$&  $3^{16}$& 1&2\\
$\langle (x^2-2)^4 +uz(x), u(x^2-2)\rangle$& $3^{14}$ &1 &2\\
\hline
    \end{tabular}
\end{table}

\begin{table}[H]
    \centering
    \ContinuedFloat
    \caption{(Continued.) Hamming and Symbol-Pair distance of $(2+v+uv)-$Constacyclic code of length 6 over $\mathbb{F}_{3} + u\mathbb{F}_{3} + v\mathbb{F}_{3} +uv\mathbb{F}_{3}$}
    \begin{tabular}{|c|c|c|c|}
\hline              
Ideal($\mathcal{C})$&$\eta_{\mathcal{C}}$  &$d_H$&$d_{sp}$\\
 \hline
 $\langle (x^2-2)^4+u(x^2-2)z(x),u \rangle$&  $3^{16}$& 1&2\\
$\langle (x^2-2)^4+u(x^2-2)z(x),u (x^2-2)\rangle$& $3^{14}$ & 1&2\\
$\langle (x^2-2)^4+u(x^2-2)z(x),u (x^2-2)^2\rangle$& $3^{12}$ & 1&2\\
\hline
       $\langle (x^2-2)^4+u(x^2-2)^2z(x),u \rangle$& $3^{16}$ & 1&2\\
$\langle (x^2-2)^4+u(x^2-2)^2z(x),u (x^2-2)\rangle$& $3^{14}$ &1 &2\\
$\langle (x^2-2)^4+u(x^2-2)^2z(x),u (x^2-2)^2\rangle$& $3^{12}$ &1 &2\\
$\langle (x^2-2)^4+u(x^2-2)^2z(x),u (x^2-2)^3\rangle$& $3^{10}$ &1 &2\\
\hline
$\langle (x^2-2)^4+u(x^2-2)^3 z(x) ,u\rangle$& $3^{16}$ &1 &2\\
$\langle (x^2-2)^4+u(x^2-2)^3 z(x) ,u (x^2-2)\rangle$&$3^{14}$ &1 &2\\
$\langle (x^2-2)^4+u(x^2-2)^3 z(x) ,u  (x^2-2)^2\rangle$& $3^{12}$ & 1&2\\
$\langle (x^2-2)^4+u(x^2-2)^3 z(x) ,u (x^2-2)^3  \rangle$& $3^{10}$ &1 &2\\
\hline
$\langle (x^2-2)^5+uz(x),u \rangle$&$3^{14}$  & 1&2\\
\hline
$\langle (x^2-2)^5+u(x^2-2)z(x),u \rangle$&$3^{14}$  &1 &2\\
$\langle (x^2-2)^5+u(x^2-2)z(x),u (x^2-2)\rangle$&$3^{12}$  & 1&2\\
\hline
$\langle (x^2-2)^5+u(x^2-2)^2z(x),u \rangle$& $3^{14}$ & 1&2\\
$\langle (x^2-2)^5+u(x^2-2)^2z(x),u (x^2-2)\rangle$&  $3^{12}$&1 &2\\
$\langle (x^2-2)^5+u(x^2-2)^2z(x),u (x^2-2)^2\rangle$&$3^{10}$  &1 &2\\
\hline
$\langle (x^2-2)^5+u(x^2-2)^3 z(x) ,u\rangle$& $3^{14}$ & 1&2\\
$\langle (x^2-2)^5+u(x^2-2)^3 z(x) ,u(x^2-2)\rangle$&$3^{12}$  & 1&2\\
$\langle (x^2-2)^5+u(x^2-2)^3 z(x) ,u(x^2-2)^2\rangle$&$3^{10}$  & 1&2\\
$\langle (x^2-2)^5+u(x^2-2)^3 z(x) ,u(x^2-2)^3\rangle$& $3^{8}$ & 1&2\\
\hline
$\langle (x^2-2)^5+u(x^2-2)^4 z(x) ,u \rangle$& $3^{14}$ & 1&2\\
$\langle (x^2-2)^5+u(x^2-2)^4 z(x) ,u (x^2-2)\rangle$&$3^{12}$  &1&2\\
$\langle (x^2-2)^5+u(x^2-2)^4 z(x) ,u (x^2-2)^2\rangle$& $3^{10}$ &1 &2\\
$\langle (x^2-2)^5+u(x^2-2)^4 z(x) ,u (x^2-2)^3\rangle$& $3^{8}$ & 1&2\\
$\langle (x^2-2)^5+u(x^2-2)^4 z(x) ,u (x^2-2)^4\rangle$& $3^{6}$ & 2&4\\
\hline
    \end{tabular}
\end{table}
\end{example}


\end{document}